\documentclass[journal,two columns]{IEEEtran}
\usepackage{amsfonts}
\usepackage{amsmath,amssymb}
\usepackage{graphicx,subfigure}
\usepackage{color}      
\usepackage{epsfig}
\usepackage{amssymb}

\newtheorem{theorem}{Theorem}[section]
\newtheorem{lemma}[theorem]{Lemma}

\newtheorem{corollary}[theorem]{Corollary}

\begin{document}

\title{An achievable region for the double unicast problem based on a minimum cut analysis}

\author{\IEEEauthorblockN{Shurui Huang, {\it{Student Member, IEEE}} and Aditya Ramamoorthy, {\it{Member, IEEE}}\\
}
\thanks{This research was supported in part by NSF grants CCF-1018148 and CNS-0721453.}
\authorblockA{Department of Electrical and Computer Engineering\\
Iowa State University, Ames, Iowa 50011\\
Email: \{hshurui, adityar\}@iastate.edu}}

\maketitle
\thispagestyle{empty}
\pagestyle{empty}
\begin{abstract}
We consider the multiple unicast problem under network coding over
directed acyclic networks when there are two source-terminal pairs,
$s_1-t_1$ and $s_2-t_2$. Current characterizations of the multiple
unicast capacity region in this setting have a large number of
inequalities, which makes them hard to explicitly evaluate. In
this work we consider a slightly different problem. We assume that
we only know certain minimum cut values for the network, e.g.,
mincut$(S_i, T_j)$, where $S_i \subseteq \{s_1, s_2\}$ and $T_j
\subseteq \{t_1, t_2\}$ for different subsets $S_i$ and $T_j$. Based
on these values, we propose an achievable rate region for this
problem based on linear codes. Towards this end, we begin by
defining a base region where both sources are multicast to both the
terminals. Following this we enlarge the region by appropriately
encoding the information at the source nodes, such that terminal
$t_i$ is only guaranteed to decode information from the intended
source $s_i$, while decoding a linear function of the other source.
The rate region takes different forms depending upon the
relationship of the different cut values in the network.
%
\end{abstract}

\section{Introduction}


The problem of characterizing the utility of network coding for
multiple unicasts is an intriguing one. In the multiple unicast
problem there is a set of source-terminal pairs in a network that
wish to communicate messages. This is in contrast to the multicast
problem where each terminal requests exactly the same set of
messages from the source nodes. The multicast problem under network
coding is very well understood. In particular, several papers
\cite{al}\cite{rm}\cite{Tracy06} discuss the exact capacity region
and network code construction algorithms for this problem.

However, the multiple unicast problem is not that well understood. A significant amount of previous work has attempted to find inner
and outer bounds on the capacity region for a given instance of a
network.
In \cite{yan06isit}, an information theoretic characterization for
directed acyclic networks is provided. However, explicit evaluation
of the region is computationally intractable for even small networks
due to the large number of constraints. The authors in
\cite{HarveyIT} propose an outer bound on the capacity region. Price et
al. \cite{javidi08} provide an outer bound on the capacity
region in a two unicast session network, and provided a network
structure in which their outer bound is the exact capacity region. The
work of \cite{Traskov06} forms a linear optimization to characterize an 
achievable rate region by packing butterfly structures in the
original graph. This approach is limited since only the XOR operation is
allowed in each butterfly structure.

In this work we propose an achievable region for the two-unicast problem using linear network codes. Our setup is somewhat different from the above-mentioned works in that we consider
directed acyclic networks with unit capacity edges and assume that
we only know certain minimum cut values for the network, e.g.,
mincut$(S_i, T_j)$, where $S_i \subseteq \{s_1, s_2\}$ and $T_j
\subseteq \{t_1, t_2\}$ for different subsets $S_i$ and $T_j$. 
This is related to the work of Wang
and Shroff \cite{wangIT10} (see also \cite{shenvi}) for two-unicast that presented a
necessary and sufficient condition on the network structure for the
existence of a network coding solution that supports unit rate
transmission for each $s_i-t_i$ pair. In this work we consider general rates.
Reference \cite{feder09} is related
in the sense that they give an achievable rate region for this
problem based on the number of edge disjoint paths for $s_i-t_i$
pair. In our work we propose a new achievable rate region given
additional information about the network resources. The work of \cite{huangicc}
considered the three unicast session problem in which each source is transmitting at unit rate. Finally, reference \cite{JafarISIT} applies the technique of interference alignment in the case of three unicast sessions and shows that communication at half the mincut of each source-terminal pair is possible.



This paper is organized as follows. Section \ref{sec:system}
introduces the system model under consideration. Section
\ref{sec:rateregion} contains the precise problem formulation and
the derivations of our proposed achievable rate region. Section \ref{sec:comparison} compares our achievable region to existing literature.
Due to space limitations, some of the lemma proofs are not given and can be found in \cite{ITWsup}.

\section{System Model}
\label{sec:system} We consider a network represented by a directed acyclic graph $G=(V,E)$. There is
a source set $S=\{s_1,s_2\}\in V$ in which each source observes a random process with a discrete integer entropy, and there
is a terminal set $T=\{t_1,t_2\}\in V$ in which $t_i$ needs to
uniquely recover the information transmitted from $s_i$ at rate
$R_i$. Each edge $e\in E$ has unit capacity and can transmit one
symbol from a finite field of size $q$. If a given edge has a higher
capacity, it can be divided into multiple parallel edges with unit
capacity. 
Without loss of generality (W.l.o.g.), we assume that there is no
incoming edge into source $s_i$, and no outgoing edge from terminal
$t_i$. By Menger's theorem, the minimum cut between sets $S_{N_1}\subseteq S$ and $T_{N_2}\subseteq T$ is
the number of edge disjoint paths from $S_{N_1}$ to $T_{N_2}$,
and will be denoted by $k_{N_1-N_2}$ where $N_1,N_2\subseteq \mathcal N=\{1,2\}$. For two unicast sessions, we define
the \textit{cut vector} as the vector of the cut values
$k_{1-1}$, $k_{2-2}$, $k_{1-2}$, $k_{2-1}$, $k_{12-1}$, $k_{12-2}$,
$k_{1-12}$, $k_{2-12}$ and $k_{12-12}$. 

The network coding model in this work is based on \cite{rm}. Assume
source $s_i$ needs to transmit at rate $R_i$. Then the random
variable observed at $s_i$ is denoted as
$X_i=(X_{i1},X_{i2},\cdots,X_{iR_i})$, where each $X_{ij}$ is an element of $GF(q)$; the $X_i$s are assumed to be independent. For linear network codes, the
signal on an edge $(i,j)$ is a linear combination of the signals on
the incoming edges on $i$ or a linear combination of the source
signals at $i$. Let $Y_{e_n}$ ($tail(e_n)=k$ and $head(e_n)=l$)
denote the signal on edge $e_n\in E$. Then,
\begin{align*}
&Y_{e_n} = \sum_{\{e_m|head(e_m)=k\}}f_{m,n}Y_{e_m}\text{ if }k\in V\setminus \{s_1,s_2\}, \text{ and }\\
&Y_{e_n} = \sum_{j=1}^{R_i}a_{ij,n}X_{ij}\text{ if }X_i\text{ is
observed at } k.
\end{align*}
The \textit{local coding vectors} $a_{ij,n}$ and $f_{m,n}$ are
also chosen from $GF(q)$. We can also express $Y_{e_n}$ as
, $Y_{e_n}=\sum_{j=1}^{R_1}\alpha_{j,n}X_{1j}+\sum_{j=1}^{R_2}\beta_{j,n}X_{2j}$.
Then the global coding vector of $Y_{e_n}$ is $[\alpha_{n},
\beta_{n}] = [\alpha_{1,n},\alpha_{2,n}, \cdots,
\alpha_{R_1,n},\beta_{1,n},\beta_{2,n}, \cdots, \beta_{R_2,n}]$. We are free to choose an appropriate value of the field size $q$.

In this work, we present an achievable rate region given
a subset of the cut values in the cut vector; namely, $k_{1-1}$, $k_{2-2}$, $k_{1-2}$, $k_{2-1}$, $k_{12-1}$, $k_{12-2}$. 
W.l.o.g, we assume there are $k_{i-ij}$ outgoing edges from $s_i$ and $k_{ij-i}$ incoming edges
to $t_i$. If this is not the case one can always introduce an
artificial source (terminal) node connected to the original source
(terminal) node by $k_{i-ij}$ ($k_{ij-i}$) edges. It can be seen that the new network
has the same cut vector as the original network.

\section{Achievable rate region for a given cut vector}
\label{sec:rateregion}
First, suppose that only $t_1$ is interested in recovering the random variables $X_1$ and $X_2$ which are observed at $s_1$ and $s_2$ respectively. 
Denote the rate from $s_1$ to $t_1$ and $s_2$ to $t_1$ as $R_{11}$ and $R_{12}$. Then the capacity region $C_{t_1}$, that is achieved by routing will be
\begin{align*}
R_{11}&\leq k_{1-1},\\
R_{12}&\leq k_{2-1},\\
R_{11}+R_{12}&\leq k_{12-1}.
\end{align*}

The capacity region $C_{t_2}$ for $t_2$  can be drawn in a similar
manner. This is shown in Fig. \ref{fig:excase2}. We also find the
boundary points $a,b,c,d$ such that their coordinates are
$a=(k_{12-1}-k_{2-1},k_{2-1}),b=(k_{1-2},k_{12-2}-k_{1-2}),c=(k_{1-1},k_{12-1}-k_{1-1}),d=(k_{12-2}-k_{2-2},k_{2-2})$. A simple achievable rate region for our problem can be arrived at by
multicasting both sources $X_1$ and $X_2$ to both the terminals
$t_1$ and $t_2$.

\begin{figure}[htbp]
\subfigure[]{\label{fig:excase2}
\includegraphics[width=42mm,clip=false, viewport=70 35 190 120]{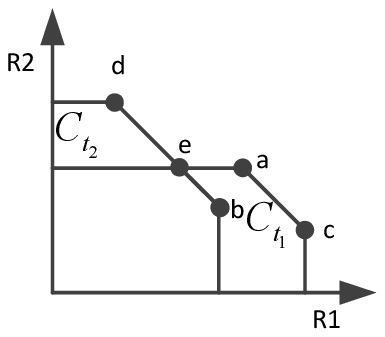}}
\subfigure[]{\label{fig:base}
\includegraphics[width=42mm,clip=false, viewport=70 35 190 120]{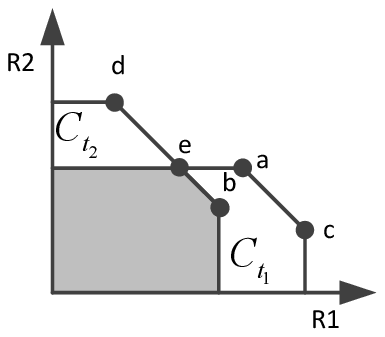}}
\caption{\label{fig:casesecond} (a) An example of a capacity region.
(b) Base region for the example. }
\end{figure}

\begin{theorem}
\label{th:main}
Rate pairs $(R_1, R_2)$ belonging to the following set $\mathcal{B}$ can be achieved for two unicast sessions.
\begin{align*}
\mathcal{B} = \{&R_1\leq \min(k_{1-2},k_{1-1}),\\
&R_2\leq \min(k_{2-1},k_{2-2}),\\
&R_1+R_2 \leq \min(k_{12-1},k_{12-2})\}.
\end{align*}
\end{theorem}

\begin{proof}
We multicast both the sources to each terminal. This can be done
using the multi-source multi-sink multicast result (Thm. 8) in \cite{rm}.
\end{proof}

Subsequently we will refer to region $\mathcal{B}$ achieved by
multicast as the \textit{base rate region} (the grey region in Fig.
\ref{fig:base}).


We now move on to precisely formulating the problem. Let $Z_i$
denote the received vector at $t_i$,  $X_i$
denote the transmitted vector at $s_i$, and $H_{ij}$ denote
the transfer function from $s_j$ to $t_i$.
Let $M_i$ denote the encoding matrix at $s_i$, i.e., $M_i$ is the transformation from $X_i$ to the transmitted symbols on the outgoing edges from $s_i$. In our formulation, we will let the length of $X_i$ to be $k_{i-i}$ (i.e., the maximum possible). For transmission at rates $R_1$ and $R_2$, we introduce precoding matrices $V_i, i = 1,2$ of dimension $R_i \times k_{i-i}$, so that the overall system of equations is as follows.
\begin{equation}
\label{eq:system}
\begin{split}
Z_1&=H_{11}M_1V_1X_1+H_{12}M_2V_2X_2,\\
Z_2&=H_{21}M_1V_1X_1+H_{22}M_2V_2X_2.
\end{split}
\end{equation}
We say that $t_i$ can receive at rate $R_i$ from $s_i$ if it can decode
$V_i X_i$ perfectly. The row dimension of the $V_i's$ can be adjusted to
obtain different rate vectors. For $(R_1, R_2) \in \mathcal{B}$, it
can be shown that there exist local coding vectors over a large enough
field such that the ranks of the different matrices in the first
column of Table \ref{table:dime} are given by the corresponding
entries in the third column, which correspond to the maximum
possible. Furthermore, by the multi-source multi-sink multicast
result, these matrices are such that
$[H_{11}M_1~~H_{12}M_2]$ is a full rank matrix of dimension
$k_{12-1}\times (R_1+R_2)$, and $[H_{21}M_1~~H_{22}M_2]$ is a full rank
matrix of dimension $k_{12-2}\times (R_1+R_2)$. In Table
\ref{table:dime}, for instance since the minimum cut between $s_1$
and $t_1$ is $k_{1-1}$, we know that the maximum rank of $H_{11}$ is
$k_{1-1}$. Using the formalism of \cite{rm}, we can conclude that
there is a square submatrix of $H_{11}$ of dimension $k_{1-1} \times
k_{1-1}$ whose determinant is not identically zero. Such appropriate
submatrices can be found for each of the matrices in the first
column of Table \ref{table:dime}. This in turn implies that their
product is not identically zero and therefore using the
Schwartz-Zippel lemma, we can conclude that there exists an
assignment of local coding vectors so that the rank of all the matrices
is simultaneously the maximum possible. For the rest of the paper,
we assume that such a choice of local coding vectors has been made. Our arguments will revolve around appropriately modifying the
source encoding matrices $M_1$ and $M_2$.

\begin{table}\scriptsize \caption{dimension of matrices} \centering
{\begin{tabular}{c|c c}
\hline\\
{\bf matrix} & $dimension $ & $rank$   \\
\hline\\
$H_{11}$ & $k_{12-1}\times k_{1-12}$ & $k_{1-1}$ \\
\hline\\
$H_{12}$ &$k_{12-1}\times k_{2-12}$ &$k_{2-1}$ \\
\hline\\
$[H_{11}~~H_{12}]$ &$k_{12-1}\times (k_{1-12}+k_{2-12})$ &$k_{12-1}$ \\
\hline\\
$M_{1}$ &$k_{1-12}\times R_1$ &$R_1$ \\
\hline\\
$H_{21}$ & $k_{12-2}\times k_{1-12}$ & $k_{1-2}$\\
\hline\\
$H_{22}$ & $k_{12-2}\times k_{2-12}$ & $k_{2-2}$ \\
\hline\\
$[H_{21}~~H_{22}]$ &$k_{12-2}\times (k_{1-12}+k_{2-12})$ &$k_{12-2}$ \\
\hline\\
$M_{2}$&  $k_{2-12}\times R_2$ & $R_2$ \\
\hline
\end{tabular}}
\label{table:dime} \vspace{-0.1in}
\end{table} \normalsize

Note that there are two boundary points of the base region (the two
boundary points may overlap). At point $Q_1$, we denote the achievable rate pair by $(R^*_1, R^*_2)$ where
\begin{align*}
R^*_1&=\min(k_{1-2},k_{1-1}), \text{~and} \\
R^*_2&=\min(\min(k_{2-1},k_{2-2}),\min(k_{12-1},k_{12-2})-R^*_1).
\end{align*}
At point $Q_2$, we denote the achievable rate pair by $(R^{**}_1, R^{**}_2)$ where
\begin{align*}
R^{**}_2&=\min(k_{2-1},k_{2-2}) , \text{~and} \\
R^{**}_1&=\min(\min(k_{1-2},k_{1-1}),\min(k_{12-1},k_{12-2})-R^{**}_2)
\end{align*}
In Fig. \ref{fig:excase2}, these boundary points are $Q_1 = b$ and $Q_2 = e$.

In what follows, we will present our arguments towards increasing the value of $R_1$ to be larger than $R^*_1$ (these arguments can be symmetrically applied for increasing $R_2$ as well). For this purpose, we will start with the point $Q_1$ and attempt to achieve points that are near it but do not belong to $\mathcal{B}$.
At $Q_1$, if
$R^*_1=k_{1-1}$, then we cannot increase $R_1$ due to the cut
constraints. Hence, we assume $R^*_1=k_{1-2}$.
Furthermore, since $k_{2-2}\geq k_{12-2}-k_{1-2} \geq
\min(k_{12-1},k_{12-2})-k_{1-2}$, $R^*_2=\min(\min(k_{2-1},k_{2-2}),\min(k_{12-1},k_{12-2})-R^*_1)=\min(k_{2-1}, \min(k_{12-1},k_{12-2})-k_{1-2})$.

In this paper we refer to $k_{1-2} + k_{2-1}$ as a measure of the interference in the network and in the subsequent discussion present achievable regions based on its value. We emphasize though that this is nomenclature used for ease of presentation. Indeed a high value of $k_{1-2}$ does not necessarily imply that there is a lot of interference at $t_2$, since the network code itself dictates the amount of interference seen by $t_2$.
The following lemma will be used extensively.
\begin{lemma}
\label{lemma:partialDecode} Consider a system of equations
$Z=H_1X_1+H_2X_2$, where $X_1$ is a vector of length $l_1$ and $X_2$
is a vector of length $l_2$ and $Z\in span([H_1~~H_2])$\footnote{Throughout the paper, $span(A)$ refers to
the column span of $A$.}. The matrix
$H_1$ has dimension $z_t\times l_1$, and rank $l_1-\sigma$, where
$0\leq\sigma\leq l_1$. The matrix $H_2$ is full rank and has
dimension $z_t\times l_2$ where $z_t\geq (l_1+l_2-\sigma)$.
Furthermore, the column spans of $H_1$ and $H_2$ intersect only in
the all-zeros vectors, i.e. $span(H_1)\cap
span(H_2)=\{0\}$. Then there exists a unique solution for
$X_2$.
\end{lemma}



\subsection{Low Interference Case}
This is the case when $k_{1-2}+k_{2-1}\leq \min(k_{12-1},k_{12-2})$. At $Q_1$, from the assumption, it follows that
$R^*_1=k_{1-2}, R^*_2=\min(k_{2-1},\min(k_{12-1},k_{12-2})-k_{1-2})=k_{2-1}$. An example is shown
in Fig. \ref{fig:rateregion}. 
Furthermore, $Q_1=Q_2 = e$.

Our solution strategy is to consider the encoding matrices $M_1$ and
$M_2$ at the point $Q_1$,
and to introduce a new encoding matrix at $s_1$, denoted $M_1'$
(with $R^*_1 + \delta$ columns) such that $span(H_{11}M'_1)\cap
span(H_{12})=\{0\}$. As shown below, this will allow $t_1$ to decode from $s_1$ at rate $R^*_1 + \delta$ and $t_2$ to decode from $s_2$ at rate $R^*_2$. After the modification, each $t_i$ is guaranteed
to decode at the appropriate rate from $s_i$. A similar argument can
then be applied for $R^*_2$ to arrive at the achievable rate region in
this case.


\begin{figure}[htbp]
\subfigure[]{\label{fig:rateregion}
\includegraphics[width=42mm,clip=false, viewport=50 57 165 140]{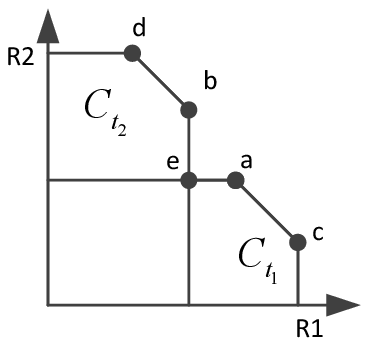}}
\subfigure[]{\label{fig:rateregion_im_case1}
\includegraphics[width=42mm,clip=false, viewport=50 57 165 140]{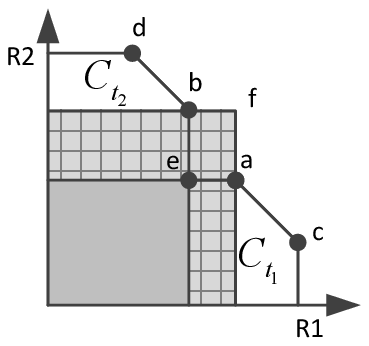}}
\caption{\label{fig:casefirst} (a) The capacity regions $C_{t_1}$ and $C_{t_2}$ for an
example of low interference case. (b) The achievable rate region for low interference case. For each point in the shaded grey area, both terminals can recover both the sources. In the hatched grey area, for a given rate point, its $x$-coordinate is the rate for $s_1-t_1$ and its $y$-coordinate is the rate for $s_2-t_2$; the terminals are not guaranteed to decode both sources in this region.}
\end{figure}

At the point $Q_1$, the rates are $R^*_1 = k_{1-2}, R^*_2 = k_{2-1}$.
Since both terminals can decode both sources, it holds that
\begin{align*} &rank(H_{i1}M_1)= k_{1-2}, rank(H_{i2}M_2)=k_{2-1}, \text{~and~}\\ &span(H_{i1}M_1)\cap span(H_{i2}M_2)=\{0\} \text{~for $i = 1,2.$} \end{align*}
By analyzing the properties of the above matrices, we have Theorem
\ref{th:caseA}. Before we state the theorem, we first give the
following lemma which will be used in
proving Theorem \ref{th:caseA}.

\begin{lemma}
\label{lemma:in} \textit{Rate Increase Lemma.} In the base region, denote the achievable rates at $Q_1$ as
$R^*_1$ and $R^*_2$, and the corresponding encoding matrices as $M_1$
and $M_2$.
Let $rank([H_{11}~~H_{12}M_2])=r\geq R^*_1+R^*_2$. 
There exist a series of full rank matrices
$\bar{M}_1^{(n)}=[\tilde{M}^{(n)}_1~~M_1]$ of dimension $k_{1-12}\times
(n+R^*_1)$ such that
$rank([H_{11}\bar{M}_1^{(n)}~~H_{12}M_2])=R^*_1+R^*_2+n$, $0\leq n \leq
(r-R^*_1-R^*_2)$.
\end{lemma}


\begin{theorem}
\label{th:caseA}
Given a cut vector, if $k_{1-2}+k_{2-1}\leq
\min(k_{12-1},k_{12-2})$, then the rate pair in the following region
can be achieved.

\noindent \underline{Region 1:}
\begin{align*}
R_1&\leq k_{12-1}-k_{2-1},\\
R_2&\leq k_{12-2}-k_{1-2},
\end{align*}
which is shown in Fig. \ref{fig:rateregion_im_case1}.
\end{theorem}

\begin{proof}
In this case, $R^*_1=k_{1-2}$ and $R^*_2=k_{2-1}$ is the boundary point
$Q_1=Q_2$. We will try to find full rank matrix $M'_1$ of dimension $k_{1-12}\times (k_{12-1}-k_{2-1})$ and full rank matrix
$M'_2$ of dimension $k_{2-12}\times (k_{12-2}-k_{1-2})$ such that
the system of equations can be written as
\begin{equation*}
\begin{split}
Z_1&=H_{11}M'_1V'_1X_1+H_{12}M'_2V'_2X_2,\\
Z_2&=H_{21}M'_1V'_1X_1+H_{22}M'_2V'_2X_2,
\end{split}
\end{equation*}
and $V'_1X_1$ can be decoded at $t_1$, $V'_2X_2$ can be decoded at
$t_2$.

First, note that $rank(H_{12} M_2) = rank(H_{12})$, which implies
that $span(H_{12})=span(H_{12}M_2)$. Therefore
$rank([H_{11}~~H_{12}]) = rank([H_{11}~~H_{12}~~H_{12}M_2 ] =
rank([H_{11}~~ H_{12}M_2])$. Together, this implies that
$rank([H_{11}~~H_{12}M_2]) = k_{12-1}$. Using the Rate Increase Lemma, we can find the matrix
$M'_1$ such that the following conditions are satisfied: (i) $M'_1$
is a full rank matrix of dimension $k_{1-12}\times
(k_{12-1}-k_{2-1})$, (ii) $rank(H_{11}M'_1)=k_{12-1}-k_{2-1}$ and
(iii) $span(H_{11}M'_1)\cap span(H_{12})=\{0\}$. (i) is from the
Rate Increase Lemma. (ii) and (iii) hold because of the following
argument. From Rate Increase Lemma and the fact that
$rank(H_{12}M_2)=rank(H_{12})=k_{2-1}$, we will have
\begin{align*}
k_{12-1}&=rank(H_{11}M'_1~~H_{12}M_2)\\
&= rank(H_{11}M'_1)+rank(H_{12}M_2)\\
&~~~~-rank(span(H_{11}M'_1)\cap span(H_{12}M_2))\\
&\leq rank(H_{11}M'_1)+rank(H_{12}M_2)\\
&\leq rank(M'_1)+rank(H_{12})\\
&=k_{12-1}-k_{2-1}+k_{2-1}=k_{12-1}.
\end{align*}
Then all the inequalities become equalities. (ii) and (iii)
are satisfied. Likewise, $M_2'$ can be found with similar conditions.

Next, since $span(H_{11}M'_1)\cap span(H_{12})=\{0\}$ and
$span(H_{12}M'_2)\subseteq span(H_{12})$, we will have
$span(H_{11}M'_1)\cap span(H_{12}M'_2)=\{0\}$. By Lemma
\ref{lemma:partialDecode} and the above three conditions, $t_1$ can
decode $V'_1X_1$ at rate $k_{12-1}-k_{2-1}$, but cannot decode
$V'_2X_2$. By a similar argument, $t_2$ can decode $V'_2X_2$ at rate
$k_{12-2}-k_{1-2}$, but cannot decode $V'_1X_1$.
\end{proof}


%
%



\vspace{-0.05in}
\subsection{High Interference Case}
This is the case when $k_{1-2}+k_{2-1}\geq \min(k_{12-1},k_{12-2})$. Recall that we also assume that $k_{1-2}\leq k_{1-1}$.
At $Q_1$, $R^*_1=k_{1-2}$, $R^*_2=\min(k_{2-1},
\min(k_{12-1},k_{12-2})-k_{1-2})=\min(k_{12-1},k_{12-2})-k_{1-2}$.
This means that
$Q_1$ and $Q_2$ are two separated points. An example is shown in
Fig. \ref{fig:excase2}.
In particular, when $C_{t_1}$ is contained in $C_{t_2}$ or vice versa, the achievable region is described by this case.

Our strategy is similar to the one for the previous case, but with
important differences. We begin with the rate vector at point $Q_1$
and then attempt to increase $R_1$. However, in this particular case
we will not be able to increase $R_2$ and in fact may need to reduce
it. This is because at point $Q_1$, we have $R^*_2 = rank(H_{12}M_2) < k_{2-1}=rank(H_{12})$, i.e.,
the encoding matrix $M_2$ is such that $rank(H_{12}M_2)$ is strictly
less than the maximum possible. Therefore, if we augment $M_2$ with
additional columns to arrive at $M_2'$, it is not possible to assert
as before that the $span(H_{11}M_1')\cap span(H_{12}M'_2)=\{0\}$. Hence,
it may be possible that $s_1$ cannot be decoded at $t_1$, (after augmenting $M_2$ to $M'_2$). In this situation, we have the following result.

\begin{theorem}
\label{th:caseB} Given a cut vector, if $k_{1-2}+k_{2-1}\geq
\min(k_{12-1},k_{12-2})$ and $k_{1-2}\leq k_{1-1}$, then the rate
pair in the following region can be achieved.

\noindent \underline{Region 2:}
\begin{align*}
R_1&\leq k_{1-1},\\
R_2&\leq \min(k_{12-1},k_{12-2})-k_{1-2},\\
R_1+R_2&\leq rank([H_{11}~~H_{12}M_2]).
\end{align*}
\end{theorem}
Note that in the above characterization, the sum rate constraint depends on  $rank([H_{11}~H_{12}M_2])$; we show a lower bound on $rank([H_{11}~H_{12}M_2])$ in \ref{sec:lb}. The following lemma that discusses situations in which rates can be traded off between the two unicast sessions is needed for the proof of Thm. \ref{th:caseB}.


\begin{lemma}
\label{lemma:onetoone}  \textit{Rate Exchange Lemma.} Given that
$ rank([H_{11}M_1~~H_{12}M_2])=rank([H_{11}~H_{12}M_2])=r$, where
$M_1$ is a full rank matrix of dimension $k_{1-12}\times (r-R_2)$,
$M_2$ is a full rank matrix of dimension $k_{2-12}\times R_2$. If
$M'_1=[~\vec{\alpha}~~M_1~]$ where $\vec{\alpha}$ is a vector of
length $k_{1-12}$ and $rank(H_{11}M'_1)=r-R_2+1$, then there exists
an $M'_2$ such that $span(H_{11}M'_1)\cap span(H_{12}M'_2)=\{0\}$
where $M'_2$ is a full rank submatrix of $M_2$ of dimension $k_{2-12}\times
(R_2-1)$.
\end{lemma}
{\it Proof of Theorem \ref{th:caseB}.} Given that $k_{1-2}+k_{2-1}\geq
\min(k_{12-1},k_{12-2})$ and $k_{1-2}\leq k_{1-1}$, we will extend
the rate region from $Q_1$ where $R^*_1=k_{1-2}$,
$R^*_2=\min(k_{12-1},k_{12-2})-k_{1-2}$. At $Q_1$, we need to increase
$R_1$ while keeping $R_2$ as large as possible. By the Rate Increase Lemma, we can achieve the rate point
$R'_1=rank([H_{11}~~H_{12}M_2])-R^*_2$,
$R'_2=R^*_2$. The corresponding encoding
matrices are $M'_1$ and $M_2$.
When we want to further increase $R'_1$, we could use the Rate Exchange Lemma repeatedly. Hence, when $R'_1$ is increased by $\delta$, $R'_2$ is decreased by $\delta$
where $0\leq\delta\leq
\min(R^*_2,k_{1-1}-R'_1)$ ($\delta\leq
k_{1-1}-R'_1$ comes from the fact that $R'_1$ can be increased to at
most $k_{1-1}$).
Terminal $t_1$ can decode messages both from $s_1$ at rate
$R''_1=R'_1+\delta$ and $s_2$ at rate $R''_2=R'_2-\delta$.
Denote the new set of encoding matrices as $M''_1$ and $M''_2$.

At $t_2$, because $M''_2$ is a submatrix of $M_2$,
$span(H_{22}M''_2)\subseteq span(H_{22}M_2)$. Furthermore, we have
$span(H_{21}M''_1)\subseteq span(H_{21})=span(H_{21}M_1)$, since
$R^*_1 = k_{1-2}$. Hence, from the above argument, we will have
$span(H_{21}M''_1)\cap span(H_{22}M''_2)=\{0\}$ since
$span(H_{21}M_1)\cap span(H_{22}M_2)=\{0\}$. Then by Lemma
\ref{lemma:partialDecode}, we can decode at $R''_2=R'_2-\delta$
from $s_2$, but not decode any messages from $s_1$. \endproof

A similar analysis for $Q_2$ allows us to increase $R_2$, resulting in the following extended region.


\begin{corollary}
Given a cut vector, if $k_{1-2}+k_{2-1}\geq \min(k_{12-1},k_{12-2})$
and $k_{2-1}\leq k_{2-2}$, then the rate pair in the following
region can be achieved.

\noindent \underline{Region 3:}
\begin{align*}
R_1&\leq \min(k_{12-1},k_{12-2})-k_{2-1},\\
R_2&\leq k_{2-2},\\
R_1+R_2&\leq rank([H_{21}M_1~~H_{22}]).
\end{align*}
\end{corollary}
The overall rate region is the convex hull of base region, Region 2
and Region 3 which is shown in Fig. \ref{fig:case2im}, where boundary segment $d-f$ is achieved via timesharing.


We note that the idea of increasing one rate while decreasing the
other can also be applied to the region obtained in low interference case. Since
$rank([H_{11}~H_{12}M_2])=k_{12-1}$ and
$rank([H_{21}M_1~H_{22}])=k_{12-2}$, we can obtain the following two
new regions for low interference case.

%

\noindent
\begin{align*}
\text{\underline{Region 2':}} & ~& \text{\underline{Region 3':}} &~\\
R_1&\leq k_{1-1}   &  R_1&\leq k_{1-2}\\
R_2&\leq k_{2-1}   &  R_2&\leq k_{2-2}\\
R_1+R_2&\leq k_{12-1} & R_1+R_2&\leq k_{12-2}
\end{align*}


Finally, the achievable rate region for low interference case is the convex hull of
the region 1, 2' and 3' shown in Fig.
\ref{fig:rate_region3}, where the boundary segment $d-f$ and $f-c$ is achieved via timesharing.

\subsubsection{Lower bound of $rank([H_{11}~H_{12}M_2])$}
\label{sec:lb}
Next, we investigate the lower bound of
$rank([H_{11}~~H_{12}M_2])$. In the following argument, $R^*_1$ and
$R^*_2$ denote the rate at boundary point $Q_1$, and $M_1$ and $M_2$
denote the corresponding encoding matrices. First note that
$rank([H_{11}~H_{12}M_2])\geq rank(H_{11})=k_{1-1}$ and
$rank([H_{11}~H_{12}M_2])\geq rank([H_{11}M_1~H_{12}M_2])=R^*_1+R^*_2$.
Next we will also find another nontrivial lower bound of
$rank([H_{11}~~H_{12}M_2])$ by the following lemma.
\begin{lemma}
\label{lemma:span} Given $rank([H_{11}~H_{12}])=k_{12-1}$,
$rank(H_{12})=k_{2-1}$ and $rank([H_{12}M_2])=l$, we have
$rank([H_{11}~H_{12}M_2])\geq k_{12-1}-k_{2-1}+l$.
\end{lemma}
\begin{proof}
By the assumed conditions, there are $k_{2-1}$ columns in $H_{12}$
that are linearly independent, and in $H_{11}$, we can find a subset
of $k_{12-1}-k_{2-1}$ columns denoted $H_{11}'$ such that
$span(H_{11}')\cap span(H_{12})=\{0\}$ and
$rank(H_{11}')=k_{12-1}-k_{2-1}$, which further imply that
$rank([H_{11}'~H_{12}])= k_{12-1}$.


Since $span(H_{12}M_2)\subseteq span(H_{12})$ this means that
$span(H_{11}')\cap span (H_{12}M_2)=\{0\}$. Then
$rank([H_{11}'~H_{12}M_2])=rank(H_{11}')+rank(H_{12}M_2)-0=k_{12-1}-k_{2-1}+l$.
Hence, $rank([H_{11}~H_{12}M_2])\geq
rank([H_{11}'~H_{12}M_2])=k_{12-1}-k_{2-1}+l$.
\end{proof}

Together with the two lower bounds above, we have
$rank([H_{11}~H_{12}M_2])\geq \max(k_{1-1}, k_{12-1}-k_{2-1}+R^*_2,
R^*_1+R^*_2)$. A case where $\max(k_{1-1}, k_{12-1}-k_{2-1}+R^*_2,
R^*_1+R^*_2)=k_{12-1}-k_{2-1}+R^*_2$ is shown in Fig. \ref{fig:case2im}.

\begin{figure}[htbp]
\subfigure[]{\label{fig:case2im}
\includegraphics[width=42mm,clip=false, viewport=70 35 190 120]{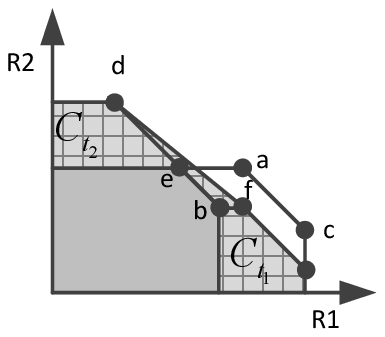}}
\subfigure[]{\label{fig:rate_region3}
\includegraphics[width=41mm,clip=false, viewport=55 60 165 140]{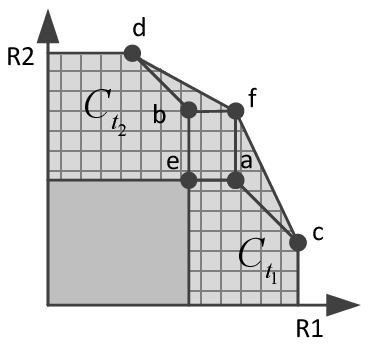}}
\caption{\label{fig:casefinal} (a) The extended rate region for high interference case. (b) The final extended rate region for low interference case. For each point in the shaded grey area, both terminals can recover both the sources. In the hatched grey area, for a given rate point, its $x$-coordinate is the rate for $s_1-t_1$ and its $y$-coordinate is the rate for $s_2-t_2$; the terminals are not guaranteed to decode both sources in this region. }
\vspace{-0.2in}
\end{figure}
\vspace{-0.1in}
\section{Comparison with existing results}
\label{sec:comparison}
The authors in \cite{wangIT10} and \cite{shenvi} explore the case when each source
transmits one symbol at a time, or equivalently, $R_1=R_2=1$ in detail, whereas we allow arbitrary rate pairs. Reference \cite{feder09}, also consider the scenario where the rates are arbitrary. 
Assuming that $k_{2-2}\leq k_{1-1}$, the basic region in \cite{feder09} is
\noindent\text{\underline{Region EF09:}}
\begin{align*}
R_1+2R_2&\leq k_{1-1}\\
R_2&\leq k_{2-2}
\end{align*}
They also extend the region using the knowledge of $k_{1-2}$, $k_{2-1}$ and other cut conditions arising from the network topology (see section IV of \cite{feder09}). 
A comparison between our region and theirs indicates that there are example networks where there exist rate points that belong to our region but not to Region EF09. Conversely, there are instances of networks where points that belong to Region EF09, do not fall within our region. The work of \cite{feder09} can be interpreted in part as an interference nulling scheme, and in future work it may be possible to incorporate this within our approach. The work of \cite{javidi08} considers several different cuts defined in the graph and propose an outer bound for the network capacity. Moreover, they provide certain network structures where the outer bound is tight. Since our work deals with an inner bound, it is qualitatively different. Finally Das et al. \cite{JafarISIT} have used interference alignment for the case of three unicast sessions, and are able to achieve a rate that is half the mincut for each unicast session. While this is an interesting result for a harder problem, the case of two unicast sessions considered here is different since each connection has only one interferer and the alignment problem does not exist. Moreover, achieving half the mincut for each session can be trivially achieved by timesharing in our problem. In that sense a comparison between our results and theirs is not possible.

\vspace{-0.1in}
\bibliographystyle{IEEEtran}
\bibliography{unicast}

\end{document}